\documentclass[usenatbib,useAMS,referee]{biom}
\def\bSig\mathbf{\Sigma}

\usepackage[figuresright]{rotating}
\usepackage{lineno,hyperref,tikz,amsmath,multirow,amssymb,ulem} 
\usepackage{algorithm}
\usepackage{algorithmic}

\makeatletter
\let\c@lofdepth\relax
\let\c@lotdepth\relax
\let\c@subfigure\relax
\let\l@subfigure\relax
\let\c@subtable\relax
\let\l@subtable\relax
\let\@listsubcaptions\relax
\let\@dottedxxxline\relax
\let\subfloat@label\relax
\let\sf@@sub@label\relax

\let\TRUE\relax
\makeatother
\usepackage{bm}
\usepackage{graphicx}
\usepackage{float}
\usepackage{subfig}

\usepackage{multirow,threeparttable}   
\usepackage{fancybox}
\usepackage{hyperref} 

\title[ZIPG-SK]{Variable Selection for Multi-Source Count Data with Controlled False Discovery Rate}

\author{
Shan Tang$^{1}$,
Shanjun Mao$^{2}$,
Shourong Ma$^{2,\dagger}$, and 
Falong Tan$^{2,\dagger}$\email{falongtan@hnu.edu.cn} \\
$^{1}$School of Mathematics and Statistics, Wuhan University of Technology, 430070, Wuhan, China\\
$^{2}$Department of Statistics and Data Science, Hunan University, 410006, Changsha, China 
}

\begin{document}

\date{{\it Received October} 2007. {\it Revised February} 2008.  {\it
Accepted March} 2008.}

\pagerange{\pageref{firstpage}--\pageref{lastpage}} 
\volume{64}
\pubyear{2008}
\artmonth{December}

\doi{10.1111/j.1541-0420.2005.00454.x}

\label{firstpage}

\begin{abstract}
The rapid generation of complex, highly skewed, and zero-inflated multi-source count data poses significant challenges for variable selection,  particularly in biomedical domains like tumor development and metabolic dysregulation. To address this, we propose a new variable selection method, Zero-Inflated Poisson-Gamma Simultaneous Knockoff (ZIPG-SK), specifically designed for multi-source count data. Our method leverages a Gaussian copula based on the Zero-Inflated Poisson-Gamma (ZIPG) distribution to construct knockoffs that properly account for the properties of count data, including high skewness and zero inflation, while effectively incorporating covariate information. This framework enables the detection of common features across multi-source datasets with guaranteed false discovery rate (FDR) control. Furthermore, we enhance the power of the method by incorporating e-value aggregation, which effectively mitigates the inherent randomness in knockoff generation. Through extensive simulations, we demonstrate that ZIPG-SK significantly outperforms existing methods, achieving superior power across various scenarios. We validate the utility of our method on real-world colorectal cancer (CRC) and type 2 diabetes (T2D) datasets, identifying key variables whose characteristics align with established findings and simultaneously provide new mechanistic insights. \\
\end{abstract}

\begin{keywords}
    E-value Aggregation, Gaussian Copula, Multi-source Count Data, Variable Selection, ZIPG-SK 
\end{keywords}

\maketitle

\vspace{-6pt}
\section{Introduction} \label{s:intro}

\vspace{-6pt}
Advances in sequencing technology have enabled the collection of large-scale, high-dimensional data from diverse sources, presenting both significant opportunities and analytical challenges in fields like microbiology and genetics \citep{ko2022metagenomics,larson2023clinician}. These datasets, exemplified by gut microbiome and single-cell sequencing, often contain a far greater number of variables than samples, with only a small subset truly relevant to phenotypic outcomes. This high-dimensionality underscores the crucial challenge of accurate variable selection \citep{candes2018panning}. The complexity is further compounded by the need to integrate data from multiple independent sources to uncover shared biological mechanisms. For instance, studies on the gut microbiome require combining data from different geographic cohorts to investigate its relationship with diseases like colorectal cancer (CRC) \citep{liu2022multi}. 
In addition to the inherent high-dimensional nature of sequencing data, many datasets, such as gut microbiome and single-cell sequencing data, are frequently represented as count data, which raises unique statistical challenges. 

Existing approaches for variable selection in such settings can be broadly grouped into two categories. The first involves direct analysis of count data, where a common strategy in single-cell RNA sequencing is to cluster cells and then perform statistical tests for differential gene expression to identify marker genes \citep{petegrosso2020machine}. For example, \citet{heumos2023best} applied the Clustifyr package for cell classification and subsequently used the Wilcoxon rank-sum test for differential gene detection. The second category involves transforming count data into compositional data prior to selection. \citet{lin2014variable} proposed a regularization method based on a log-contrast model, which has been applied in microbiome studies to uncover associations such as between body mass index (BMI) and gut microbial composition. Despite their usefulness, both approaches have limitations. Differential expression analysis after clustering may suffer from “double dipping”, i.e., using the same data for clustering and testing, which can inflate false discoveries \citep{song2023clusterde}, while compositional transformations may introduce bias through log-ratio conversion and the handling of zeros, potentially obscuring sequencing depth information. Beyond the count-data challenge, the issue of multi-source integration further complicates variable selection. A common approach to variable selection across multiple sources is to either pool datasets into a single analysis or intersect the variable sets obtained from individual sources \citep{chen2020ensemble}. However, pooling may amplify biases caused by heterogeneity across cohorts \citep{mary2007biases}, while intersecting can be overly conservative and exclude relevant features \citep{lee2019confounder}. These challenges underscore the urgent need for more robust methods tailored to variable selection in multi-source sequencing count data.

The knockoff method has emerged as a powerful approach for high-dimensional variable selection, constructing a set of “knockoff” variables to effectively control the false discovery rate (FDR) \citep{barber2015controlling}. This approach has been applied to microbiome data; for example, \citet{monti2022robust} proposed a two-step robust knockoff filter for compositional microbiome data that controls the FDR even in the presence of contamination. Model-X knockoff methods have been developed for high-dimensional settings \citep{candes2018panning}. While the original knockoff framework is based on second-order statistics, many Model-X implementations inherit gaussian assumptions, making them unsuitable for skewed and zero-inflated count data \citep{burger2020robust}. Alternative strategies, such as constructing knockoffs under assumed distributions (e.g., negative binomial) with Gaussian copula-based fitting \citep{song2023clusterde}, partially address these challenges but do not fully account for zero inflation or multi-source integration. Simultaneous Knockoffs (SK) \citep{dai2023false} further extend Model-X framework to multi-source settings, enabling detection of shared signals, but they rely on the same gaussian-based assumptions and limitations. Therefore, developing knockoff-based methods that can simultaneously handle the unique characteristics of count data and integrate information from multiple sources remains an important direction for future research.

In this paper, we propose a knockoff-based framework, Zero-Inflated Poisson-Gamma Simultaneous Knockoff (ZIPG-SK), for variable selection in multi-source count data. The method’s key contributions are threefold: (i) it models the highly skewed and zero-inflated nature of count data using a Zero-Inflated Poisson-Gamma (ZIPG) distribution \citep{jiang2023flexible} and generates knockoffs via a Gaussian copula, incorporating relevant covariates to improve statistical power; (ii) it builds on the SK framework to integrate knockoffs from multiple independent sources, adapting the method for zero-inflated count data; and (iii) it enhances FDR control and statistical power by repeating knockoff construction and incorporating e-values \citep{ren2024derandomised}, which mitigates randomness in knockoff generation. We theoretically establish the control of the FDR within our proposed framework. Through extensive simulations, ZIPG-SK demonstrates superior performance compared with existing variable selection methods for count data across diverse scenarios. Finally, we validate the approach on two real-world multi-source datasets from CRC and type 2 diabetes (T2D) studies, identifying taxa and genes consistent with existing literature while uncovering novel candidates for further biological investigation.

The remaining content of the paper is organized as follows. Section~\ref{s:model} presents the proposed model, including the Methods and Algorithm for ZIPG-SK. Section~\ref{s:theory} establishes the theoretical foundations, detailing the extension of FDR control for count data variable selection from a single-source to a rigorous multi-source framework. The main results from simulations and real data analysis are demonstrated in Section~\ref{s:simulations} and Section~\ref{s:real_data}, respectively. Section~\ref{s:conclusion} concludes the paper. A detailed description of the ZIPG-SK method, all proofs of the theoretical results, and additional experimental analyses of simulations and real-world data is provided in the Supplementary Materials.

\vspace{-12pt}
\section{Methods} \label{s:model}
\vspace{-6pt}

This section provides the construction of ZIPG-SK method for variable selection with multi-source count data. First, we introduce the foundation of knockoff-based methods in the context of multi-source multiple testing problems. Next, we describe the construction of knockoff variables for count data from individual sources. We then present complete ZIPG-SK procedure, which integrates information across all sources. Finally, we detail the aggregation algorithm for ZIPG-SK, which stabilizes the final selection results using e-values.

\vspace{-10pt}
\subsection{The multiple testing for multi-source data} 
\vspace{-6pt}

The ZIPG-SK method provides a general framework for testing the union null hypotheses on conditional associations between outcomes and a set of candidate count features. We assume data are collected from $K$ independent experiments (sources), indexed by the set $[K]=\{1,\dots,K\}$. Within the $k$-th experiment, we observe a sample of $n_k$ individuals, where each observation $(Y_{k}^{i},{W}_{k}^{i})$ is an independent and identically distributed (i.i.d.) realization from a distribution $\mathcal{D}_{k}$, i.e., $(Y_{k}^{i},{W}_{k}^{i}) \stackrel{i.i.d.}{\sim } \mathcal{D}_{k}$, for $i \in [n_k]$. Here, $Y_{k}^{i}$ denotes the outcome, and ${W}_{k}^{i}$ is a $p$-dimensional vector of count features for the $i$-th individual in experiment $k$. 

\vspace{-3pt}
For each feature $j \in [p]$ and experiment $k\in [K]$, we define the null hypothesis $H_{0,kj}$ that the $j$-th feature is conditionally independent of the outcome in the $k$-th experiment, given all other features \citep{dai2023false}. Specifically,
\vspace{-6pt}
$$ \vspace{-6pt} H_{0,kj}: {W}_{kj}\bot {Y}_k |{W}_{k,-j},\vspace{-6pt} $$ 
where ${Y}_k$ is the outcome, ${W}_{kj}$ is the $j$-th feature random variable from experiment $k$, and ${W}_{k,-j} := \{{W}_{k1},\dots, {W}_{kp}\} \setminus {W}_{kj}$ denotes the vector of all feature random variables except the $j$-th. Denote $\mathcal{H}_{0,k} =\{j \in [p] : H_{0,kj} \text{ is true}\}$. Our primary interest is to test the union null hypotheses for each feature across all experiments, defined as
\vspace{-6pt}
\begin{equation}
    H_{0j}  = \cup_{k=1}^K H_{0,kj} ,\; \text{for}\; j \in [p]. 
    \vspace{-3pt}
\end{equation}
We denote the set of false null hypotheses as $\mathcal{S} = \{ j \in [p]: H_{0j} \text{ is false}\}$ and the set of true null hypotheses as $\mathcal{H}_0 = \mathcal{S}^c = \cup_{k=1}^K \mathcal{H}_{0,k} =\{j \in [p]: H_{0,j} \text{ is true}\}$. The aim is to develop a selection procedure that returns a selection set $\widehat{\mathcal{S}} \in [p]$ with a controlled FDR, which is the expected false discovery proportion (FDP):
\vspace{-6pt}
\begin{equation}
    \begin{aligned}
        {FDR}(\widehat{\mathcal{S}})=\mathbb{E}[{FDP}(\widehat{\mathcal{S}})]=\mathbb{E}\left(\frac{|\widehat{\mathcal{S}}\cap\mathcal{H}_{0}|}{|\widehat{\mathcal{S}}|\vee1}\right).
    \end{aligned}
\end{equation}

\vspace{-6pt}
\subsection{Knockoff construction for individual-source count data} \label{Synthetic} 
\vspace{-6pt}

The construction of knockoff variables for individual-source count data is akin to synthetic null generation, and it begins with distributional assumptions for each count feature. While models such as the Zero-inflated Poisson (ZIP) and Zero-inflated negative binomial (ZINB) are commonly used for single-cell and gut microbiome data \citep{chen2022distribution}, they often fail to account for the influence of other covariates (e.g., clinical factors or dietary patterns). Moreover, these models may not adequately capture the inherent variability and heterogeneity in count data, where the relationship between mean abundance and its dispersion can vary significantly across different features.

\vspace{-3pt}
To address these limitations, we assume each count feature follows a Zero-inflated Poisson-Gamma (ZIPG) distribution, which is a flexible framework that models both mean abundance and its dispersion through covariates \citep{jiang2023flexible}. This distributional assumption, along with the incorporation of relevant covariates, allows for a more accurate representation of the data. The count data can then be effectively approximated by a multivariate distribution using a Gaussian copula.
Next, we detail the construction of knockoff variables for single-source count data. Without loss of generality, we omit the experiment subscript $k$ in the following discussions.

1. The null model: ZIPG specified by the Gaussian copula
\vspace{-6pt}

The null model assumes that the random vector of features for each individual, once transformed, follows a $p$-dimensional Gaussian distribution. The intuition is that all individuals are drawn from a homogeneous population, where each feature's marginal count distribution is a ZIPG distribution, and the correlation structure among features is captured by a Gaussian copula.
Under the null model, each count value $W_{j}^i$, representing the count of feature $j$ for individual $i$, is assumed to independently follow the hierarchical ZIPG distribution, i.e.,
$W_j^i \sim ZIPG(\mathrm{W};\lambda_j^i,\theta_j^i,\pi_j),\;j \in [p]$. Specifically,
\begin{equation}
    \begin{aligned}
        W_{j}^i\mid U_{j}^i
        &\sim 
         \begin{cases} 0 & \text{with probability } \pi_j \\ 
         \text{Poisson}(\lambda_{j}^i U_{j}^i) & \text{with probability } 1-\pi_j,
        \end{cases} \\  
        U_{j}^i  & \sim \text{Gamma} \left( (\theta_{j}^i)^{-1}, \theta_{j}^i \right),
    \end{aligned}
    \vspace{-6pt}
\end{equation} 
where $\lambda_{j}^i$ represents the true abundance level of count feature $j$ for individual $i$. The parameter $\pi_j$ denotes the zero-inflation probability, representing the probability that feature $j$ is truly zero. Additionally, $U_{j}^i$ follows a Gamma distribution with shape parameter $(\theta_{j}^i)^{-1}$ and scale parameter $\theta_{j}^i$, ensuring a unit mean and a variance of $\theta_{j}^i$ for the latent random effect. The parameters $\lambda_{j}^i$ and $\theta_{j}^i$ are referred to as the abundance mean parameter and the abundance dispersion parameter, respectively.

\vspace{-3pt}
The random vector of features for an individual, $W^i = (W_{1}^i,\dots,W_{p}^i)^{\top} $, is assumed to follow a multivariate ZIPG distribution specified by a Gaussian copula. This assumption implies that after a suitable transformation, the resulting variables are sample drawn from a $p$-dimensional Gaussian distribution with mean vector $0$ and covariance matrix $\mathbf{R}$. Specifically, \vspace{-32pt}
 \begin{equation}
    \left( \Phi^{-1}(F_{1}^i(W_{1}^i)),\dots,\Phi^{-1}(F_{p}^i(W_{p}^i))\right)^{\top} \:\stackrel{i.i.d.}{\sim }\:N_{p}\left(\mathbf{0},\mathbf{R}\right), 
    \vspace{-6pt}
\end{equation}  
where $F_j^i$ and $\Phi$ denote the cumulative distribution function (CDF) of $ZIPG(\lambda_j^i, \theta_j^i, \pi_j)$ and the standard Gaussian distribution $N(0,1)$, respectively.

\vspace{-3pt}
Crucially, the validity of the knockoff framework does not require the null model to perfectly characterize the true joint distribution under the alternative hypothesis. Instead, the null model is employed as a constructive tool to generate knockoff variables that satisfy two essential conditions: (i) exchangeability with the original variables under the null hypothesis, and (ii) conditional independence from the outcome given the original variables. By modeling the marginal ZIPG distribution of each feature, our method produces knockoff variables that preserve the zero-inflated and skewed characteristics of the observed data, thereby ensuring comparability and rigorous control of the FDR.

2. Fitting the null model to real data
\vspace{-6pt}

The null model is specified by the feature parameters $\mathbf{\Omega}=\{\lambda_j^i, \theta_j^i,\pi_j\}_{j=1}^p$ and the covariance matrix $\mathbf{R}$. To account for systematic variation and increase the precision of parameter estimates, we explicitly incorporate covariate information into the model.
Let $\mathbf{X}$ denote the covariates, e.g., clinical and dietary factors in gut microbiome studies.
By linking the ZIPG mean and dispersion parameters to the covariates via suitable functions, we enhance the quality of the generated knockoffs, which in turn contributes to higher statistical power for detecting true signals. Following \cite{jiang2023flexible}, these links are defined as follows:
\vspace{-6pt}
\begin{equation}
    \begin{aligned}
        g(\lambda_{j}^i)&=\beta_{j,0}+({X}^{i})^{\top}\bm{\beta}_{j}+\log(M^{i}),  \\
        g^{*}(\theta_{j}^i)&=\beta_{j,0}^{*}+({X}^{i})^{\top}\bm{\beta}_{j}^{*},
    \end{aligned}
    \vspace{-6pt}
\end{equation}  
where ${X}^i \in\mathbb{R}^{d}$ is the vector of covariates for the individual $i$, and $\log(M^{i})$ accounts for variations in sequencing depth with $M^{i} = \sum_{j=1}^{p} W_{ij}$ being the total read count (library size) for the $i$-th observation. In this paper, we adopt a logarithmic link function, $g(\cdot )=g^*(\cdot )=\log(\cdot )$, to ensure parameter positivity and integrate the covariate effects in a linear framework for estimation.
Notably, the model remains valid even in the absence of covariates ($\mathbf{X}^i=\mathbf{0}$), in which case it reduces to the non-covariate variant of the ZIPG-SK model (ZIPG-SK-nonx).
The parameters $\mathbf{\Omega}=\{\lambda_j^i, \theta_j^i,\pi_j\}_{j=1}^p$ are estimated via the Expectation-Maximization (EM) algorithm. The covariance matrix $\mathbf{R}$ is subsequently obtained from the sample covariance matrix of the latent gaussian variables. A more detailed description of the procedure is provided in Supplementary Materials S1.1.

3. Sampling from the fitted null model
\vspace{-6pt}

The generation of synthetic null data is a two-step process. We first independently sample $n$ gaussian vectors from the estimated $N_p(\mathbf{0}, \hat{\mathbf{R}})$, denoted as $(\widetilde{V}_{1}^i,\dots,\widetilde{V}_{p}^i)$ for $i \in [n]$. These vectors are then converted to ZIPG count vectors by inverting the transformation process: 
\vspace{-6pt}
\begin{equation}
    \begin{aligned}
        \widetilde{{W}}^{i}
        &:= \left( (\widehat{F}_{1}^{i})^{-1}(\Phi(\widetilde{V}_{1}^i)), \dots, (\widehat{F}_{p}^i)^{-1}(\Phi(\widetilde{V}_{p}^i)) \right)^{\top},\:i \in [n], 
    \end{aligned}
    \vspace{-6pt}
\end{equation} 
where $\widehat{F}_j^i$ is the estimated CDF of the marginal ZIPG distribution for feature $j$. The resulting matrix $\widetilde{\mathbf{W}} = (\widetilde{{W}}^1, \dots, \widetilde{{W}}^n)^{\top}$ represents the synthetic null data. It is important to note that the rows of $\widetilde{\mathbf{W}}$ are independently sampled from the fitted null model and do not have a one-to-one correspondence with the rows of the real data matrix $\mathbf{W}$.

\vspace{-6pt}
\subsection{ZIPG-SK for multi-source count data}\label{multi-source} 
\vspace{-6pt}
For each experiment $k \in [K]$, we construct a vector of test statistics, ${Z}_k \in \mathbb{R}^p$, for the original features matrix $\mathbf{W}_k$, and a corresponding vector, $\widetilde{{Z}}_k \in \mathbb{R}^p$, for their knockoff copies $\widetilde{\mathbf{W}}_k$. These statistics are constructed to satisfy the swap property: if a feature ${W}_{k,j}$ is swapped with its knockoff $\widetilde{{W}}_{k,j}$, then the statistics $Z_{k,j}$ and $\widetilde{Z}_{k,j}$ are also swapped.
We consider two primary approaches for constructing these statistics:
\vspace{-6pt}
\begin{itemize}
    \item Model-X Knockoffs (MX): 
    This approach directly leverages a model to predict the outcome $Y_k$ from the data $(\mathbf{W}_k,\widetilde{\mathbf{W}}_k)$. Examples include absolute coefficients from penalized generalized linear regression or feature importance from random forests.
    \vspace{-3pt}
    \item Cluster-Based Differential Expression Tests (ClusterDE): Statistics are obtained using a cluster-informed DE framework \citep{song2023clusterde}.
\end{itemize}
\vspace{-6pt}

For each experiment $k \in [K]$, we first compute per-experiment filter statistics ${C}_k = (C_{k1}, \dots, C_{kp})^{\top}$ from the original and knockoff test statistics $[{Z}_k, \widetilde{{Z}}_k]$ using a flip sign function $f_{FS}(\cdot)$, e.g., the difference function $C_{kj} = f_{FS}(Z_{kj},\widetilde{Z}_{kj}) = Z_{kj} - \widetilde{Z}_{kj}, \;j \in [p]$. To combine the evidence across all $K$ experiments, we synthesize the per-experiment statistics into a single vector ${C} \in \mathbb{R}^p$ using a One Swap Flip Sign Function (OSFF): 
\vspace{-6pt}
\begin{equation}\label{C1}
    C=f_{OSFF} ([{Z}_1,\widetilde{{Z}}_1],\dots,[{Z}_K,\widetilde{{Z}}_K]).
    \vspace{-6pt}
\end{equation} 
The OSFF is guaranteed by construction to satisfy the required sign-flipping property \citep{dai2023false}. A commonly used OSFF is the Direct Difference Hadamard Product, which computes the element-wise product of the difference vectors from each experiment: ${C} = \odot_{k=1}^K({Z}_k-\widetilde{{Z}}_k)$, where $\odot$ denotes the Hadamard product. More detailed information on the OSFF is provided in Supplementary Materials S1.2.

\vspace{-6pt}
Using the filter statistics ${C}$, the knockoff/knockoff+ filter is applied to obtain the selection set $\widehat{S} \text{ or }\widehat{S}_+$ under the SK/SK+ procedure. 
With the knockoff filter, the ZIPG-SK algorithm provides the selection set $\widehat{S} = \{ j: C_{j}\geq \tau\}$, where 
\vspace{-6pt}
\begin{equation}
    \tau= \min\left \{ t\in {\mathcal{C}}_{+ }: \frac {\#\{ j: C_{j}\leq - t\} }{\#\{ j: C_{i}\geq t\} \vee 1}\leq q\right \} .
    \label{S1}
    \vspace{-3pt}
\end{equation} 
With a more conservative knockoff+ filter, the ZIPG-SK algorithm provides the selection set $\widehat{S} _{+ }= \{ j: C_{j}\geq \tau_{+ }\}$, where  
\vspace{-6pt}
\begin{equation}
    \tau_{+ }= \min \left \{ t\in {\mathcal{C}}_{+ }: \frac {1+ \#\{ j: C_{j}\leq - t\} }{\#\{ j: C_{j}\geq t\} \vee 1}\leq q\right \}.
    \label{S2}
    \vspace{-3pt}
\end{equation} 
Here, $q$ is the target FDR level and $\mathcal{C}_+=\{|C_j|:|C_j|> 0\}$.
A detailed description of the ZIPG-SK algorithm is provided in Algorithm~\ref{alg:01}.
\vspace{-6pt}
\begin{algorithm}[h] 
    \caption{The Algorithm for ZIPG-SK.}
    \label{alg:01}
    \begin{algorithmic}[1]
        \REQUIRE Count matrices $\bm{W}_k \in \mathbb{R}^{n_k \times p}$, labels $\bm{Y}_k \in \mathbb{R}^{n_k}$, covariates $\bm{X}_k \in \mathbb{R}^{n_k \times d}$ for $k \in [K]$.
        \ENSURE Synthetic null data ${\widetilde{\bm{W}}_k},\;k \in [K]$; The filter statistics ${C}$; The selection set $\widehat{S} \text{ or }\widehat{S}_+$. 
        \vspace{-30pt}
        \FOR{each $k \text{ in } \{1, 2, \dots, K\}$}
        \STATE Synthetic null generation: $ \widetilde{\bm{W}}_k $ based on $\bm{W}_k,\; \bm{X}_k$ (see Section~\ref{Synthetic});
        \STATE Test statistics: ${Z}_k,{\widetilde{Z}}_k$ are calculated by Model-X knockoffs or ClusterDE (see Section~\ref{multi-source});
        \ENDFOR
        \STATE Filter statistics: select the appropriate OSFF function to calculate ${C}=f_{OSFF}([{Z}_1,\widetilde{{Z}}_1],\dots,[{Z}_K,\widetilde{{Z}}_K])$ (see equation~(\ref{C1})).
        \STATE Threshold calculation and feature selection: use the filter statistics $\mathbf{C}$ and apply the knockoff / knockoff+ filter to obtain $\widehat{S} \text{ or }\widehat{S}_+$ (see equation~(\ref{S1}) and ~(\ref{S2}));\\
        \RETURN $\bm{\widetilde{W}_k},\;k \in [K]$; ${C}$; $\widehat{S} \text{ or }\widehat{S}_+$.
    \end{algorithmic}
\end{algorithm}
 
\vspace{-12pt}
\subsection{Aggregation Algorithm for ZIPG-SK}  
\vspace{-6pt}

The knockoffs-based framework, while guaranteeing FDR control for variable selection, is susceptible to selection instability due to the inherent randomness of the feature generation process. To mitigate this stochastic instability and enhance the reproducibility of the ZIPG-SK method while retaining statistical power, we develop a derandomized aggregation strategy for test/fliter statistics. 
Specifically, inspired by the work of \citet{ren2024derandomised}, we introduce an e-value–based approach for aggregation.
An e-value, E, is a non-negative random variable satisfying $\mathbb{E}[E] \le 1$ under the null hypothesis, a property that allows for robust aggregation across independent tests by taking their average \citep{vovk2021values}. This enables the construction of a more robust and derandomized version of ZIPG-SK.

\begin{algorithm}[!htbp ]  
    \caption{Aggregating Algorithm for ZIPG-SK (Agg-ZIPG-SK).}
    \label{alg:02}
    \begin{algorithmic}[1]
        \setlength{\abovedisplayskip}{3pt}
        \setlength{\belowdisplayskip}{3pt}
        \REQUIRE Count matrices $\bm{W}_k \in \mathbb{R}^{n_k \times p}$, labels $\bm{Y}_k \in \mathbb{R}^{n_k}$, covariates $\bm{X}_k \in \mathbb{R}^{n_k \times d}$ for $k \in [K]$; $B \in \mathbb{N}^+$ copies; $\alpha \in (0,1)$.
        \ENSURE The Aggregating statistics ${E}^{Z}$ or ${E}^C$; The selection set $\hat{S}_{\text{Agg}}$.
        \FOR{each $b \text{ in } \{1, 2, \dots, B\}$}
        \FOR{each $k \text{ in } \{1, 2, \dots, K\}$}
        \STATE Synthetic null generation: $\bm{\widetilde{W}}^{(b)}_k$ based on $\bm{W}_k,\; \bm{X}_k$;
        \STATE 
        Test statistics: ${Z}^{(b)}_k,{\widetilde{Z}}^{(b)}_k$ are calculated by Model-X knockoffs or ClusterDE;
        \ENDFOR 
        \ENDFOR 
         \STATE \textbf{Aggregation (e-value based)}: (\textit{Underline denotes the aggregated output of $B$ copies.}) 
         \vspace{-30pt}
        \IF {\{Aggregation of \textbf{test statistics}\}} 
            \STATE $\underline{{E}^{Z}_k}$ obtained by eqution~(\ref{Evalues})-~(\ref{EZ}), 
             \STATE ${E}^{Z} = \frac{1}{K}\sum_{k=1}^K {E}^{Z}_k$; 
        \ELSE[{Aggregation of \textbf{filter statistics}}]
            \STATE ${C}^{(b)}=f_{OSFF}([{Z}^{(b)}_1,\widetilde{{Z}}^{(b)}_1],\dots,[{Z}^{(b)}_K,\widetilde{{Z}}^{(b)}_K])$,
             \STATE $\underline{{E}^C}$ obtained by eqution~(\ref{Evalues2})-~(\ref{EZ2});
        \ENDIF
         \STATE Apply the e-BH procedure (equation~(\ref{ebh})) to the aggregated $e$-values ${E}^{Z}$ or ${E}^C$ at level $\alpha$ to obtain the selection set $\widehat{S}_{\text{Agg}}$.\\
        \RETURN ${E}^{Z}$ or ${E}^C$; $\hat{S}_{\text{Agg}}$ .
    \end{algorithmic}  
\end{algorithm}
 \vspace{-3pt}

\vspace{-3pt}
The aggregation procedure, referred to as Agg-ZIPG-SK (Algorithm~\ref{alg:02}), stabilizes the final selection results by utilizing e-values. Specifically, we first generate $B$ knockoff copies 
$ \{ \widetilde{\bm{W}}_k^{(b)} \}_{b=1}^B $ for $k \in [K]$. The knockoff test statistics, $\{{Z}^{(b)}_k, {\widetilde{Z}}^{(b)}_k \}_{b=1}^B$, or the filter statistics $\{{C}^{(b)}\}_{b=1}^B$ is then calculated based on $ \{ \bm{W}_k, \widetilde{\bm{W}}_k^{(b)} \}_{b=1}^B $. 
The Agg-ZIPG-SK procedure aggregates the statistics of choice—either the $B$ replicates of knockoff test statistics, $\{{Z}^{(b)}_k, {\widetilde{Z}}^{(b)}_k \}_{b=1}^B$, or alternatively, the filter statistics $\left\{{C}^{(b)}\right\}_{b=1}^B$ directly. We now detail the implementation and specific mechanics of these two aggregation schemes.

\vspace{-3pt}
Under the first aggregation scheme, the Agg-ZIPG-SK procedure aggregates the $B$ knockoff test statistics $\{{Z}^{(b)}_k,{\widetilde{Z}}^{(b)}_k\}_{b=1}^B $ from the $k$-th source. 
Let $C_{k,j}^{(b)} = f_{FS}(Z_{k,j}^{(b)},\widetilde{Z}_{k,j}^{(b)})$, where $Z_{k,j}^{(b)}$ and $\widetilde{Z}_{k,j}^{(b)}$ are the $j$-components of ${Z}^{(b)}_k$ and ${\widetilde{Z}}^{(b)}_k$, respectively. 
Following \citet{ren2024derandomised}, for each feature $j \in [p]$ and each replicate $b \in [B]$, the e-value $e_{k,j}^{(b)}$ is defined as:
\vspace{-3pt}
\begin{equation}
    \begin{aligned} 
        e^{(b)}_{k,j} =  \frac{p \cdot \mathbf{1}\{ C^{(b)}_{k,j} \geq T_{C,k}^{(b)} \}}{1 + \sum_{l \in [p]} \mathbf{1}\{ C^{(b)}_{k,l} \leq - T_{C,k}^{(b)} \}},
    \end{aligned}
    \label{Evalues}
    \vspace{-3pt}
\end{equation}
where $T_{C,k}^{(b)}$ is the replicate-specific threshold, obtainable analogously to equation~\eqref{S2} with $C_j$ replaced by $C_{k,j}^{(b)}$.
This yields $B$ corresponding e-values, $\{{e}^{(b)}_{k,j}\}_{b=1}^B$. For each feature $j \in [p]$, we aggregate the e-values from these $B$ knockoff replicates by taking the average:
\vspace{-3pt}
\begin{equation}
    \begin{aligned} 
        {E}^{Z}_{k,j} = \frac{1}{B} \sum_{b=1}^{B} {e}^{(b)}_{k,j}.
    \end{aligned}
    \label{EZ}
    \vspace{-3pt}
\end{equation}
However, the knockoff-derived e-values $e^{(b)}_{k,j}$ in~\ref{Evalues} are relaxed e-values: they satisfy the global bound $\sum_{j\in \mathcal{H}_0}\mathbb{E}[e^{(b)}_{k,j}]\le p$ rather than $\mathbb{E}[e^{(b)}_{k,j}]\le 1$ for each null $j \in \mathcal{H}_0$. \citet{ren2024derandomised} show that averaging these $B$ relaxed e-values across knockoff replicates preserves the aggregated bound, $\sum_{j\in \mathcal{H}_0}\mathbb{E}[{E}^{Z}_{k,j}]\le p$, which is sufficient to guarantee exact finite-sample FDR control. 
Subsequently, the $K$ source-specific aggregated e-values $E^{Z}_{k,j}$ are combined by taking the mean to obtain the final $e$-value $E_{j}^{Z},\;j\in[p]$. 
The final selection set, $\widehat{S}_{\text{Agg}}$, is obtained by applying the e-BH procedure at level $\alpha$ to the aggregated $e$-values $E_{1}^{Z}, \ldots, E_{p}^{Z}$. The rejection set is determined by the rule:
\vspace{-3pt}
\begin{equation}
    \begin{aligned} 
        S_{\text{AGG}} = \left\{ j : E_j \geq \frac{p}{\alpha \widehat{l}} \right\}, 
    \end{aligned}
    \label{ebh}
    \vspace{-3pt}
\end{equation}
where $ \widehat{l} = \max \{ l \in [p] : E_{(l)} \geq \frac{p}{\alpha l}\} $, $E_{(1)} \geq \dots \geq E_{(p)}$ are the ordered $e$-values, and $\widehat{l}=0$ if the set is empty. This procedure guarantees $\text{FDR} \leq \alpha$ under arbitrary dependence among the $E_j$'s \citep{wang2022false}.
While the ClusterDE methods calculate the filter statistic based on $p$-values, an e-value can be calibrated into a valid $p$-value via the admissible $e_{value}$-to-$p_{value}$ transformation, $p_{value}=\min(1,1/E)$ \citep{vovk2021values}. Therefore, the filter statistic can also be computed by aggregating the corresponding $p$-values across the $K$ sources using a suitable aggregation technique. 

\vspace{-3pt}
In the second scheme, the aggregation relies on the filter statistics $\left\{{C}^{(b)}\right\}_{b=1}^B$, where the vector ${C}^{(b)} = (C^{(b)}_{1}, C^{(b)}_{2}, \dots, C^{(b)}_{p} )^{\top}$. Based on the filter statistics $\{{C}^{(b)}\}_{b=1}^B$, the e-value $e^{(b)}_j$ for the $j$-th feature in the $b$-th replicate is defined as:
\begin{equation}
    \begin{aligned} 
        e^{(b)}_j =  \frac{p \cdot \mathbf{1}\{ C^{(b)}_j \geq T_C^{(b)} \}}{1 + \sum_{l \in [p]} \mathbf{1}\{ C^{(b)}_l \leq - T_C^{(b)} \}},
    \end{aligned}
    \label{Evalues2}
    \vspace{-3pt}
\end{equation}
where $C_j^{(b)}$ is the knockoff filter statistic and $T_C^{(b)}$ is the replicate-specific threshold defined analogously to equation~\eqref{S2} with $C_j$ replaced by $C_{j}^{(b)}$. Consequently, the e-values also satisfy $\sum_{j\in \mathcal{H}_0}\mathbb{E}[e^{(b)}_{j}]\le p$, thereby maintaining valid FDR control. We then obtain the aggregated e-value ${E}^{C}_{j}$ by taking the average over the $B$ replicates:
\vspace{-6pt}
\begin{equation}
    \begin{aligned} 
        {E}^{C}_{j} = \frac{1}{B} \sum_{b=1}^{B} {e}^{(b)}_{j}.
    \end{aligned}
    \label{EZ2}
    \vspace{-3pt}
\end{equation} 
Likewise, the corresponding selection set $S_{\text{AGG}}$ is obtained from the final e-values ${E}^{C}$ using the e-BH procedure.

\vspace{-12pt}
\section{Theoretical results} \label{s:theory}
\vspace{-6pt}

The theoretical development parallels the model construction process introduced earlier, beginning with the individual-source setting and extending to the multi-source case. As noted by \citet{dai2023false}, selecting the intersection of variable sets across $K$ data sources via intersection formula fails to guarantee FDR control. To overcome this limitation, we build upon the established definition of the Model-X knockoffs framework introduced by \citet{candes2018panning} and utilize the SK method proposed by \citet{dai2023false}. This joint formulation enables rigorous FDR control in the challenging multi-source setting.

\vspace{-6pt}
\subsection{Theoretical results for the individual-source data} 
We begin by establishing the concept of null features under conditional independence. 
\vspace{-12pt}
\begin{definition}[Null features]
    A feature $W_j$ is defined as ``null'' if and only if response $Y$ is independent of $W_j$ conditional on other features, ${W}_{-j} = \{W_1, \ldots, W_p\} \setminus \{W_j\}$. The set of null features is denoted $\mathcal{H}_0 \subset [p]$, and a feature $W_j$ is called ``non-null'' or relevant if $j \notin \mathcal{H}_0$.
\end{definition}
\vspace{-6pt}
 
The knockoffs framework aims to identify the largest possible subset of conditional non-null variables while ensuring valid FDR control. In this work, the feature vector ${W}$ consists of count variables. For this data type, the knockoff copy $\tilde{W}_j$ is generated by sampling from the conditional distribution $P(W_j \mid {W}{-j})$, where by sampling we explicitly mean drawing a new realization using a random number generator. In particular, the construction is implemented through a Gaussian copula approach tailored for ZIPG features. This guarantees the fundamental conditional independence property:
\vspace{-6pt}
$$
\tilde{W}_j \mid ({W}_{-j}, Y) \stackrel{d}{=} W_j \mid {W}_{-j},
\vspace{-3pt}
$$ 
where $\overset{d}{=}$ denotes equality in distribution.
The knockoff construction for the individual source immediately satisfies the fundamental properties required by the Model-X framework \citep{candes2018panning}. The pairwise exchangeability between null count features and their knockoffs is thus immediate.
\vspace{-6pt}
\begin{lemma}[Pairwise exchangeability]\label{exchangeability}
    For any subset $S\subset\mathcal{H}_0$ of null indices,
    \vspace{-6pt}
    $$({W},\tilde{{W}}) \mid {Y} \overset{d}{=} ({W},\tilde{{W}})_{\mathrm{swap}(S)} \mid {Y}, $$
    where the vector $({W},\tilde{{W}})_{\mathrm{swap}(S)}$ is obtained from $({W},\tilde{{W}})$ by swapping the entries $W_j$ and $\tilde{{W}}_j$ for each $j \in S$.
\end{lemma}
\vspace{-6pt}

This property implies that the null features can be exchanged with their knockoffs without altering the joint distribution of $({W}, \tilde{{W}})$ given ${Y}$. As a direct consequence, the feature test statistics $Z_j$ satisfy the coin flipping property. This implies that under the null hypothesis, the probability of the statistic being positive or negative is equal. The proof of Lemma~\ref{exchangeability} is provided in Supplementary Material S2.1.

\vspace{-6pt}
\begin{lemma}[Coin flipping]
    \label{signs1}
    Let $C= f_{FS}(Z,\widetilde{Z})$. Conditional on the absolute values $(|C_{1}|,\ldots,|C_{p}|)$, the signs of the null $C_{j}$'s, $j\in\mathcal{H}_0$, are i.i.d. coin flips.
\end{lemma}   
\vspace{-6pt}

Since the knockoff variables constructed for the individual-source case satisfy both the pairwise exchangeability (Lemma \ref{exchangeability}) and coin flipping (Lemma \ref{signs1}) properties, the theoretical guarantees for FDR  control follow directly from \citet{candes2018panning}. The detailed proof of Lemma~\ref{signs1} is provided in Supplementary Material S2.1. The theoretical focus of this study, therefore, lies in extending these FDR guarantees to the more intricate problem of variable selection across multiple data sources.

\vspace{-6pt}
\subsection{Theoretical results for the multi-data source case}
\vspace{-6pt}

In the multi-source setting, the test statistics ${Z}_k$ and $\tilde{{Z}}_k$ are derived separately for each source $k$ from $( \mathbf{W}_k, \tilde{\mathbf{W}}_k )$. According to Subsection~\ref{multi-source} and Lemma~\ref{exchangeability}, these test statistics exhibit the swap property: swapping the feature and knockoff variables $({W}_k, \tilde{{W}}_k)$ results in the corresponding swap of $({Z}_k, \tilde{{Z}}_k)$. This property ensures that the multi-source test statistics $( [{Z}_1, \tilde{{Z}}_1], \ldots, [{Z}_K, \tilde{{Z}}_K])$ satisfy the joint pairwise exchangeability condition, conditional on the response ${Y}$.
The key to establishing FDR control in the multi-source setting is to show that the final filter statistics ${C}$ exhibits a specific joint sign-flipping symmetry, where ${C} = f_{OSFF}( [ {Z} _{1}, \tilde{{Z} } _{1}], \dots, [ {Z} _{K}, \tilde{{Z} } _{K}])$ given in~(\ref{C1}).
The following lemma generalizes the Coin Flipping Property (Lemma $\ref{signs1}$) to the multi-source setting. Let ${\epsilon} \in\{\pm1\}^{p}$ be a sign sequence, independent of $C$, where $\epsilon_{j}= +1$ for all $j\not\in \mathcal{H}_0$ and $\epsilon _{j}$ are independent random variables that are uniformly distributed on $\{\pm 1\}$ for all $j \in \mathcal{H}_0$.  

\vspace{-8pt}
\begin{lemma}[Joint sign-flipping symmetry] \label{signs}
Let ${C} = f_{OSFF}( [ {Z} _{1}, \tilde{{Z} } _{1}], \dots, [ {Z} _{K}, \tilde{{Z} } _{K}])$. Then
\vspace{-6pt}
$$(C_{1},...,C_{p})\stackrel{d}{=}(C_{1}\cdot\epsilon_{1},...,C_{p}\cdot\epsilon_{p}).$$
\end{lemma}  
\vspace{-12pt}
 
This joint symmetry property ensures that, for all null indices $j \in \mathcal{H}_0$, the selection probability of the original variable $W_j$ equals that of its knockoff counterpart $\tilde{W}_j$. This is essential for constructing selection procedures that control the FDR. The validity of the resulting selection procedure ZIPG-SK is formalized by the following theorem. 

\vspace{-6pt}
\begin{theorem}[FDR control]\label{FDR2} 
Suppose that the joint sign-flipping symmetry holds for the filter statistics $C$. For a given target $\text{FDR}$ level $q$, the selection procedures, resulting in $\widehat{S}$ (knockoff filter) and $\widehat{S}_+$ (knockoff+ filter), control the modified and standard FDR, respectively:
\vspace{-12pt}
\begin{align*}
    \mathrm{mFDR} = \mathbb{E}\left[\frac{|\widehat{S}\cap\mathcal{H}_{0}|} {|\widehat{S}|+1/q}\right] \leq q, \quad  \quad
    \mathrm{FDR} = \mathbb{E} \left [ \frac {| \widehat{S}_{+} \cap \mathcal{H}_{0} | }{| \widehat{S}_{+} | \vee 1}\right ] \leq q,
\end{align*} 
where $\mathcal{H}_{0}$ denotes the set of indices corresponding to variables under the joint null hypothesis.
\end{theorem}
\vspace{-6pt}

The $\mathrm{mFDR}$ is often used for computational convenience, but in high-dimensional settings, the difference between $\mathrm{mFDR}$ and the standard FDR is generally negligible. While the knockoff+ method provides stricter FDR control and is more conservative, the standard knockoff filter is more widely favored in practice due to its increased power and broad applicability. The  proof of Theorem~\ref{FDR2} is provided in Supplementary Material S2.2.

\vspace{-3pt}
The following theorem formalizes the FDR control guarantee provided by applying the e-BH procedure to the aggregated e-values.
\vspace{-12pt}
\begin{theorem}[FDR control of aggregated e-values]\label{FDR3} 
Let ${E} = ({E}_1, \ldots, {E}_p)^{\top}$ be the vector of the final aggregated e-values obtained in Algorithm~\ref{alg:02}. If $\sum_{j \in \mathcal{H}_{0}} \mathbb{E}[{E}_{j}] \leq p$, then applying the e-BH procedure at level $\alpha$ yields a selection set $\widehat{S}_{\text{Agg}}$ that controls the False Discovery Rate, that is, $\text{FDR}(\widehat{S}_{\text{Agg}}) \leq \alpha $.
\end{theorem}
\vspace{-6pt}


The proof of Theorem~\ref{FDR3} adapts e-value-based derandomization \citep{ren2024derandomised} and structurally extends it to prove the rigorous composability of the necessary global expectation bound under our multi-source aggregation scheme. The detailed proof of Theorem~\ref{FDR3} is provided in Supplementary Material S2.2.

\vspace{-20pt}
\section{Simulations}\label{s:simulations}
\vspace{-6pt}
\subsection{Simulation settings}
\vspace{-6pt}
The simulation study is designed to evaluate multi-source variable selection when the response variable ${Y}_k \in \{0, 1\}^{n_k}$ is binary across $K$ distinct datasets. For each dataset $k$, the count feature matrix is denoted by $\mathbf{W}_k \in \mathbb{R}^{n_k \times p}$. Features are partitioned into signal and null sets, $\mathbf{W}_k = [\mathbf{W}_{k,S} \,\, \mathbf{W}_{k,N}]$, where $\mathbf{W}_{k,S}$ contains features significantly associated with ${Y}_k$ and $\mathbf{W}_{k,N}$ contains irrelevant features. The response ${Y}_k$ is modeled as a binary outcome corresponding to two balanced groups, e.g., ${Y}_k = ({0}_{n_k/2}, {1}_{n_k/2})$. The generation process ensures that signal features exhibit distinct marginal distributions between the two response labels, while null features are drawn from the same distribution across all individuals. The sequencing depth for sample $i$ is defined as $M^i = \sum_{j=1}^p W_{ij}^k$. Additionally, an associated covariate matrix $\mathbf{X} \in \mathbb{R}^{n_k \times d}$ is generated from $\mathbf{X} \sim N_d (\mathbf{0},\mathbf{\Sigma})$, where $\mathbf{\Sigma} = (\rho_{ij})_{d \times d}$ has a weak correlation structure defined by $\rho_{ij} = 0.1^{|i-j|}$. Detailed parameters and the specific simulation procedure for count data generation are provided in Supplementary Material S3.1.

\vspace{-6pt}
The proposed multi-source method is benchmarked against two types of comparison procedures. First, we compare knockoff construction models: Model-X knockoff (using a simple gaussian baseline) and scDesign2 \citep{sun2021scdesign2}. scDesign2 is a transparent simulator based on the copula framework that generates count data while accurately capturing variable correlations (supporting marginal distributions such as Poisson, NB, and ZINB). We employ scDesign2's underlying model to construct knockoffs variables for our count data. Second, we compare two common strategies distinct from the simultaneous method: Pooling, where all $K$ datasets are merged for a single selection, and Intersection where variables are selected individually per dataset and the final discovery set is the intersection of results.

\vspace{-12pt}
\subsection{Simulation results}
\vspace{-6pt}
We first define the configurations used in the simulation. Three methods for knockoff construction are compared: Model-X knockoff (MX), a baseline using the Gaussian distribution; scDesign2 (ZINB), which relies on the ZINB distribution; and the proposed ZIPG model specified by the Gaussian copula (ZIPG). These are combined with three variable selection strategies: Intersection (Inter), Pooling (Pool), and the Simultaneous Knockoff (SK). We also use three test statistics—clusterDE (DE), glmnet (GLM), and random forest (RF)—to compute feature importance. Therefore, the ZIPG-SK method can be further specified as ZIPG-SK-DE, ZIPG-SK-GLM, and ZIPG-SK-RF, each representing the method calculated based on a specific statistic. We set the FDR control level to $q=0.2$ and evaluate performance on $K=2$ datasets, where the signal size is 10\% of the features.

\vspace{-12pt}
\subsubsection{FDR Control and Power} 
\vspace{-6pt}
Figure~\ref{fig_var} summarizes the results for FDR control and power over 50 experiments for different methods, where each dataset has a sample size of 400 and variable dimensions $p$ range from 200 to 1000. 
Figure~\ref{fig_var} (a) compares the FDR control and power across three knockoff construction models: ZIPG-SK, MX-SK, and ZINB-SK. We observe that the ZIPG-SK method exhibits the most robust performance across varying dimensions ($p$). Its FDR curve remains stable control, consistently hovering near target level of $q=0.2$ (Figure~\ref{fig_var} (a), lower plot). In comparison, MX-SK shows a conservative over-control of the FDR, and ZINB-SK exhibits slight FDR inflation in high-dimensional settings. The power further underscore the efficacy of the proposed method (Figure~\ref{fig_var} (a), upper plot). ZIPG-SK consistently achieves the highest power across all dimensions, indicating its superior ability to detect true signals while rigorously controlling the FDR.
Figure~\ref{fig_var} (b) compares the performance of the multi-source strategies—Inter, Pool, and SK—all implemented using the ZIPG model. The ZIPG-SK approach maintains stable FDR control and superior power relative to Inter and Pool. In contrast, the Inter and Pool strategies show greater variability in FDR control, with Pool frequently failing to maintain the nominal level of $q=0.2$ (Figure~\ref{fig_var} (b)). Regarding test statistics, the GLM-based method yields the lowest power (below 0.5), significantly trailing the DE and RF methods (over 0.6). This disparity likely stems from the GLM's reliance on linear assumptions, which may poorly capture the underlying non-linear relationships in the simulated count data. Further details are presented in Figure S.1.
Overall, simulation results validate that the ZIPG-SK framework provides the most robust combination of rigorous FDR control and high detection power for multi-source count data.

\vspace{-12pt}
\begin{figure}[!htbp] 
    \centering
 \begin{center}
        \subfloat[SK methods]{
        \begin{minipage}[c]{0.8\textwidth}
            \includegraphics[width=1\linewidth]{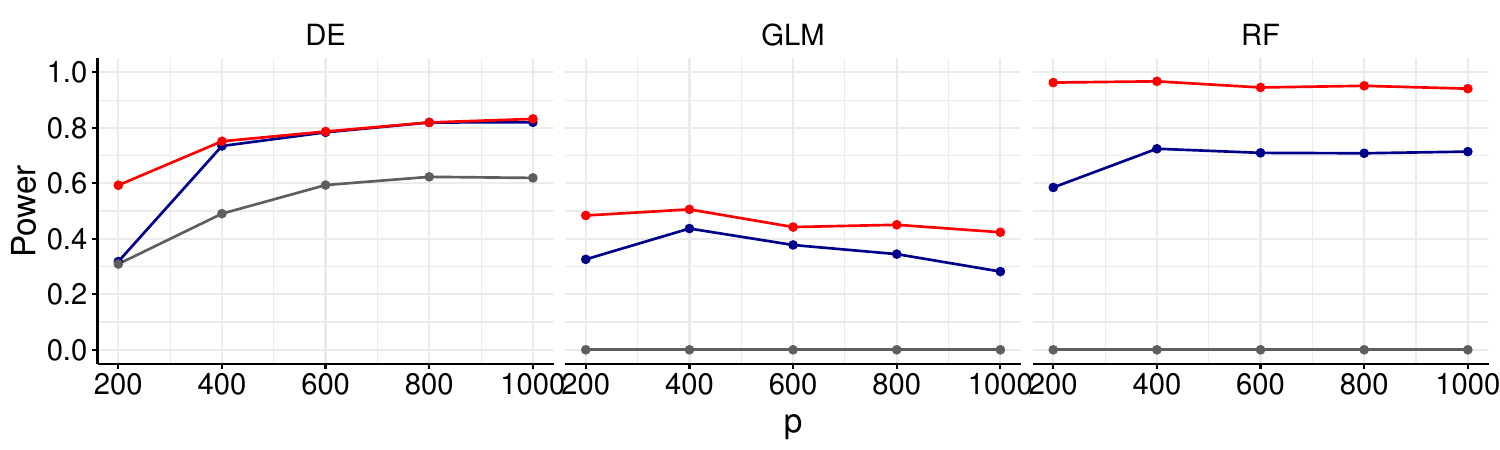} 
            \includegraphics[width=1\linewidth]{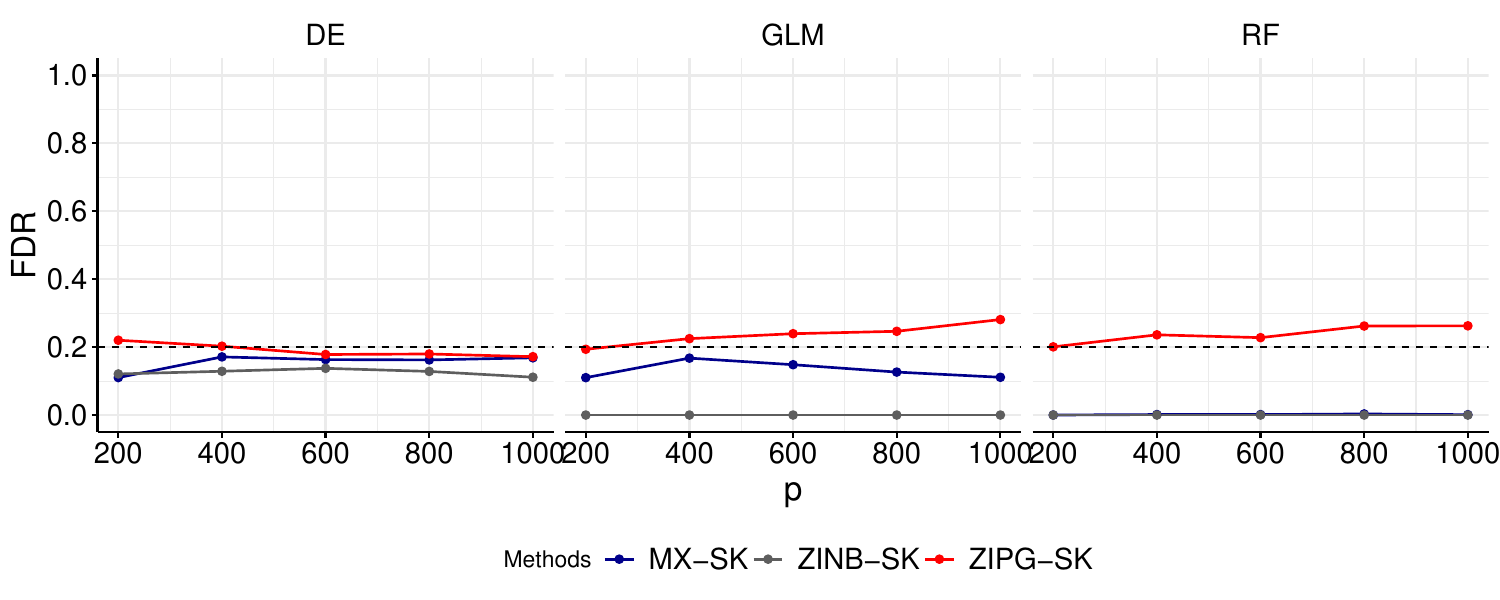}
        \end{minipage}
        }
    
        \subfloat[ZIPG methods]{
        \begin{minipage}[c]{0.8\textwidth}
            \includegraphics[width=1\linewidth]{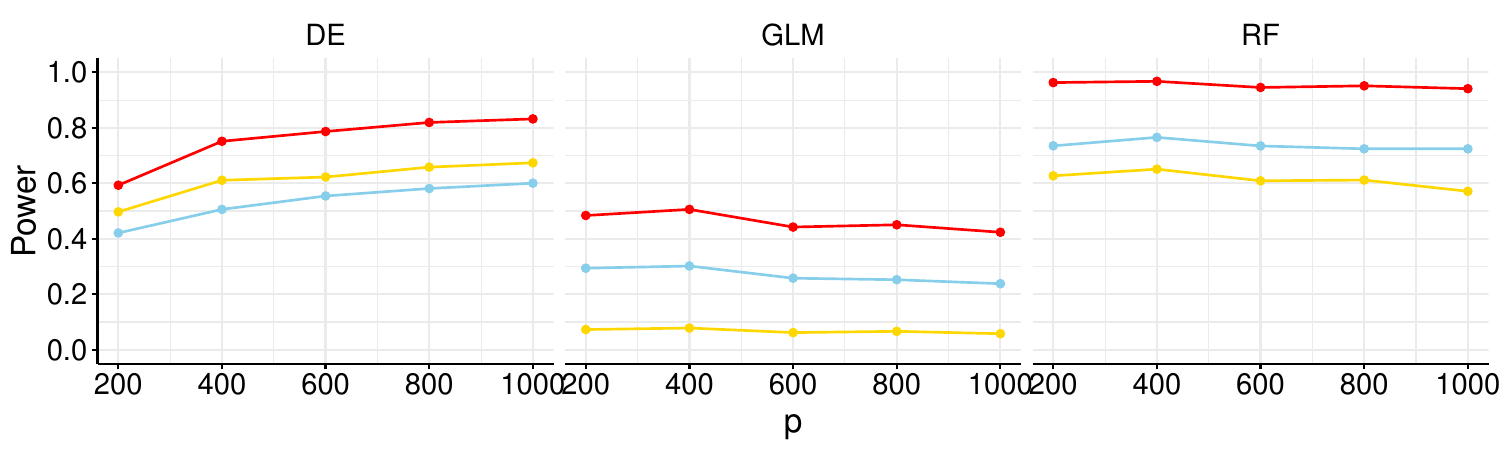} 
            \includegraphics[width=1\linewidth]{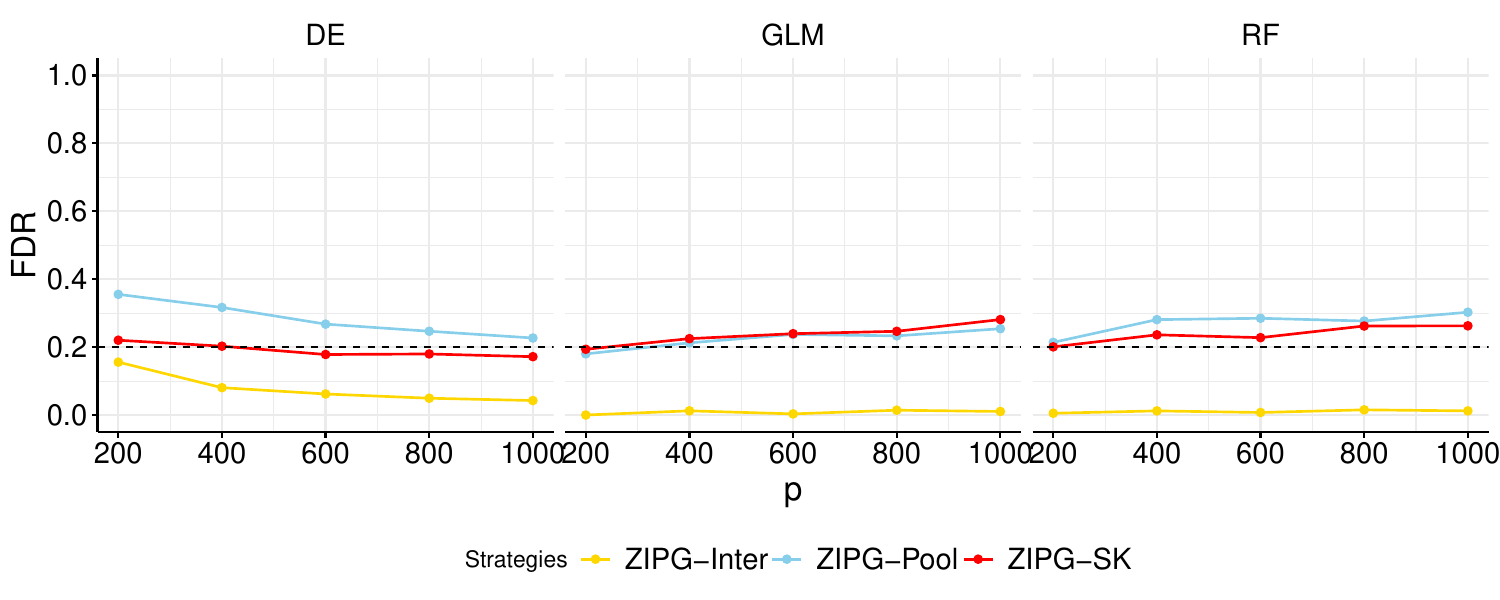}
        \end{minipage}
        }
\end{center}
    \caption{
    FDR Control and Power of Multi-Source Variable Selection Methods on Copula-based Count Data. Simulation Results on FDR Control and Power. The results are shown for $K=2$ datasets with FDR controlled at $q=0.2$. Three knockoff construction models are compared: the proposed ZIPG (using Gaussian copula), ZINB (using scDesign2), and MX (gaussian baseline). Three variable selection strategies are compared: Simultaneous Knockoff (SK), Intersection (Inter), and Pooling (Pool). Three test statistics are used: DE (clusterDE), GLM (glmnet), and RF (Random Forest). Panel (a) compares the FDR and Power of different knockoff construction methods (ZIPG, ZINB, MX) within the SK framework. Panel (b) compares the FDR and Power of various variable selection strategies (SK, Inter, Pool) based on the ZIPG method.  
    }
    \label{fig_var}
    \vspace{-6pt}
\end{figure}
\vspace{-6pt}

\vspace{-12pt}
\subsubsection{Aggregation Comparison and Robustness Analysis}
\vspace{-6pt}

We compare the ZIPG-SK method with its e-value aggregated version (Agg-ZIPG-SK) and its non-covariate variant (ZIPG-SK-nonx), the latter defined by setting $\mathbf{X} = \mathbf{0}$. 
Figure~\ref{fig_B} presents this comparative analysis of the structural variants. ZIPG-SK and ZIPG-SK-nonx methods produce similar results, likely because the ZIPG-SK model includes sequencing depth information, which may impact more than the covariates in data generation. Comparing simultaneous scenarios, Agg-ZIPG-SK achieves better FDR control than ZIPG-SK while maintaining comparable statistical power. This highlights the enhanced control achieved by the e-value-based aggregation approach without compromising detection power.

\vspace{-12pt}
\begin{figure}[!htbp]
    \centering
    \subfloat{\includegraphics[width=0.8\textwidth]{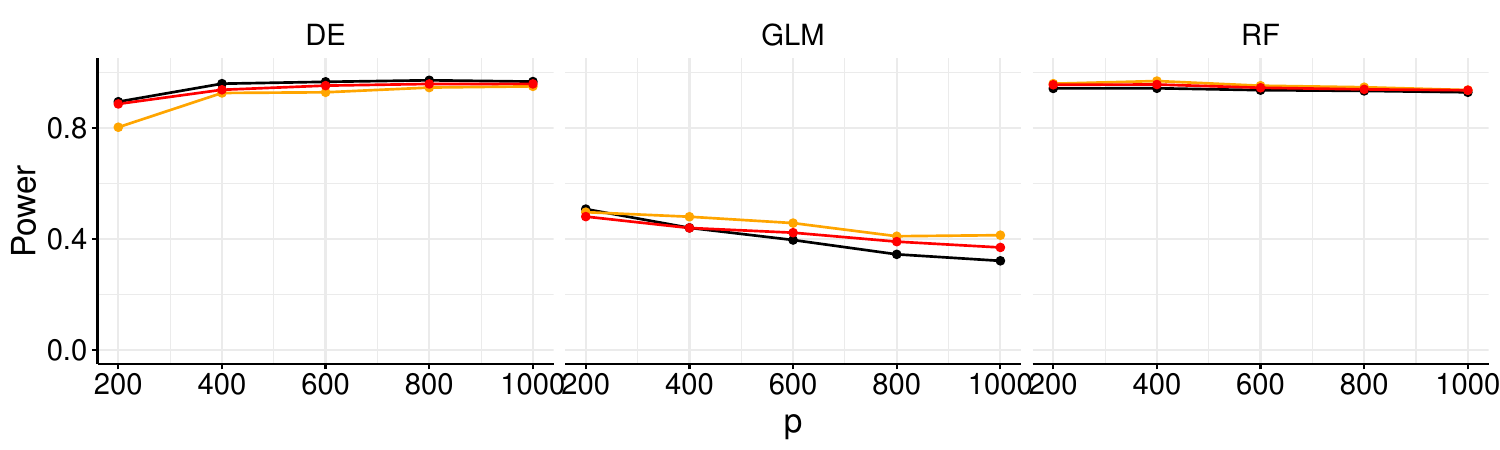}}
    
    \subfloat{\includegraphics[width=0.8\textwidth]{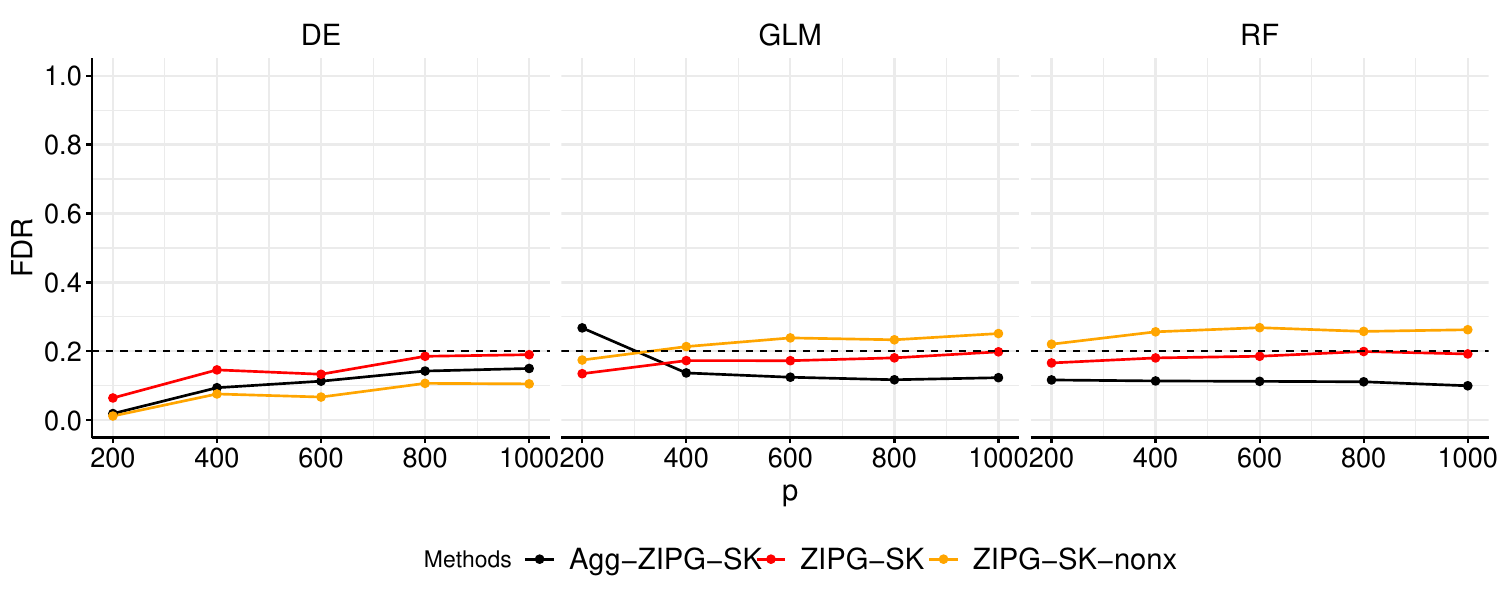}}
    \caption{
    FDR Control and Power Comparison of ZIPG-SK Variants. FDR and power results comparing the standard ZIPG-SK method with its variants: the Aggregated version (Agg-ZIPG-SK) and the non-covariate version (ZIPG-SK-nonx). The ZIPG-SK-nonx case corresponds to the special setting $\mathbf{X} = \mathbf{0}$ in the model. Results are shown across varying feature dimensions. 
    }
    \label{fig_B}
    \vspace{-6pt}
\end{figure}
\vspace{-6pt}

The study further investigates the robustness of the ZIPG-SK method against two common challenges in count data: high zero-inflation (Figure S.2) and inconsistent signal sets across sources (denoted as Diff, Figure S.3). The proposed ZIPG-SK method demonstrates stable FDR control across varying levels of zero-inflation, highlighting its robustness to the challenges posed by a high zero fraction. Although the FDR increases slightly with higher inter-source discrepancy (Diff) for all methods, ZIPG-SK consistently maintains optimal FDR control while achieving superior power across all zero-inflation and discrepancy scenarios. 

\vspace{-3pt}
The consistency of these findings is confirmed in additional settings including scenarios without a copula structure and cases with more data sources ($K=3$). Furthermore, Supplementary Materials (Section S3) provide additional details on the simulation settings including the selection of OSFF functions used in the algorithm.

\vspace{-24pt}
\section{Real data applications} \label{s:real_data}
\vspace{-6pt}


To demonstrate the effectiveness of our approach in multi-source data analysis, we focus on a complex count-based dataset——the gut microbiome dataset related to CRC. A detailed analysis for the T2D dataset, which yields consistent findings and supports the generalizability of our method, is presented in Supplementary Material S5 due to space limitations.

\vspace{-12pt}
\subsection{Datasets description}
\vspace{-6pt}

This dataset is the Gut Microbiome Data, comprising publicly available metagenomic data integrated from six distinct geographical cohorts, focusing on the association between gut microbiota and CRC \citep{liu2022multi}. The dataset includes 1014 individuals (505 CRC patients and 509 healthy controls, HC) and 845 gut microbiota species. Careful preprocessing and resulting geographical diversity make this a representative multi-source problem. 
Additional details on data preprocessing and feature-specific characteristics of the above dataset are provided in Supplementary Material S4.1.

\vspace{-12pt}
\subsection{Application to gut microbiome data}
\vspace{-6pt}

The association between gut microbiota and CRC is well-established, though significant inter-individual and geographical variations in microbial composition suggest that only a subset of microbes exhibit strong associations \citep{wong2023gut}. Our primary objective is to employ the ZIPG-SK method to identify core gut microbiota species associated with CRC across multiple geographical regions while maintaining FDR control.

\vspace{-3pt}
To ensure the stability and generalizability of the findings, we utilize a five-fold cross-validation (CV) scheme. In each iteration, four folds of data are used for variable selection, and the results from the five CV iterations are integrated to form one result set. This process is repeated 30 times to obtain comprehensive outcomes. In addition to comparing the ZIPG-SK method with the other knockoff methods (MX-SK and ZINB-SK) and selection strategies (Inter and Pool), we include two commonly used variable selection methods for compositional gut microbial data as external benchmarks. The clr-lasso \citep{susin2020variable} relies on the centered log-ratio transformation (clr) before applying LASSO regression. The coda-lasso \citep{lu2019generalized} employs a log-contrast model specifically designed for compositional data analysis (coda). These external benchmarks serve to evaluate the performance of our proposed approach against specialized compositional data methods.

\vspace{-3pt}
Figure~\ref{fig_real} illustrates the number and overlap of selected gut microbes using different methods and strategies. Figure~\ref{fig_real} (a) and (b) depict the outcomes obtained from the various methods and strategies compared in this study. It is observed that the Pool strategy generally selects the largest number of gut microbes, followed by the SK method, while the Inter strategy yields the fewest selections, sometimes resulting in zero. These results are reasonable, as the Pool strategy combines all data for variable selection, increasing statistical power, while the Inter strategy is inherently conservative. The DE approach tends to select a higher number of variables compared to the GLM and RF approaches, likely due to the nature of the Wilcoxon test-based statistics used in DE. The Inter and SK strategies select fewer gut microbes; for instance, the intersection set in the upper half of the upset plot (Figure~\ref{fig_real} (a)) contains only two variables. Figure~\ref{fig_real} (b) presents the outcomes under the SK-based strategy, showing that although the method generally selects only a few variables, there is significant overlap among the majority of the chosen gut microbes, confirming robust discovery.

\vspace{-12pt}
\begin{figure}[!htbp]
    \centering
    \subfloat[All methods]{
    \begin{minipage}[c]{0.49 \textwidth}
    \includegraphics[width=1\linewidth]{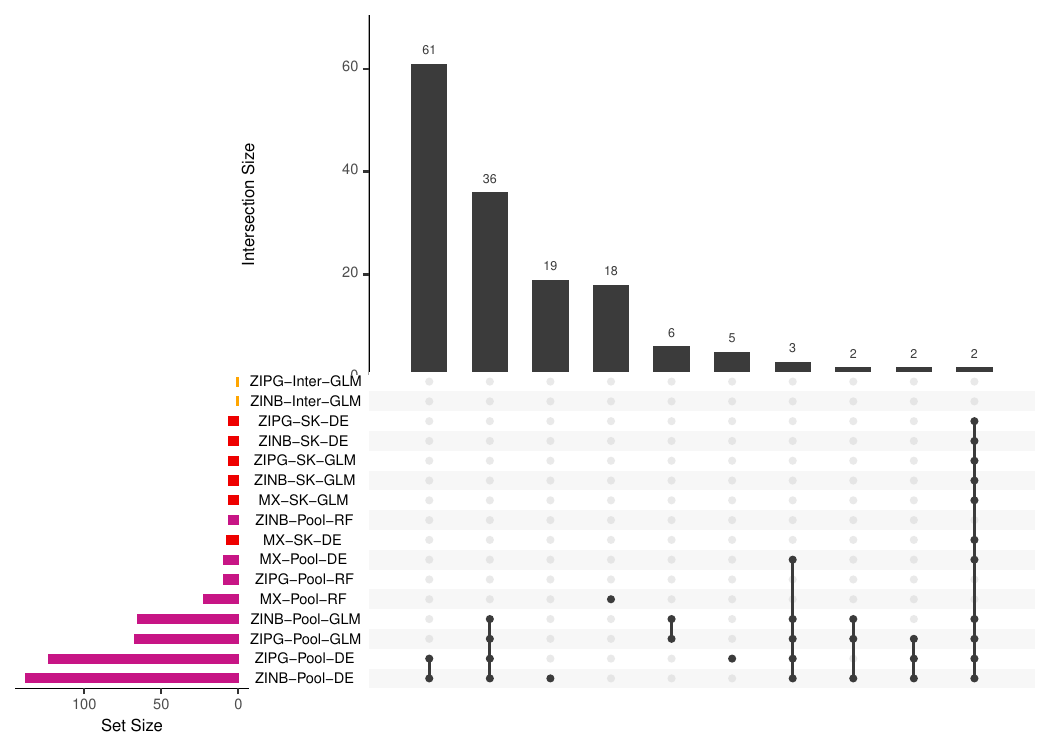}
    \end{minipage}}
    \subfloat[SK-based methods]{
    \begin{minipage}[c]{0.22\textwidth}
    \includegraphics[width=1\linewidth]{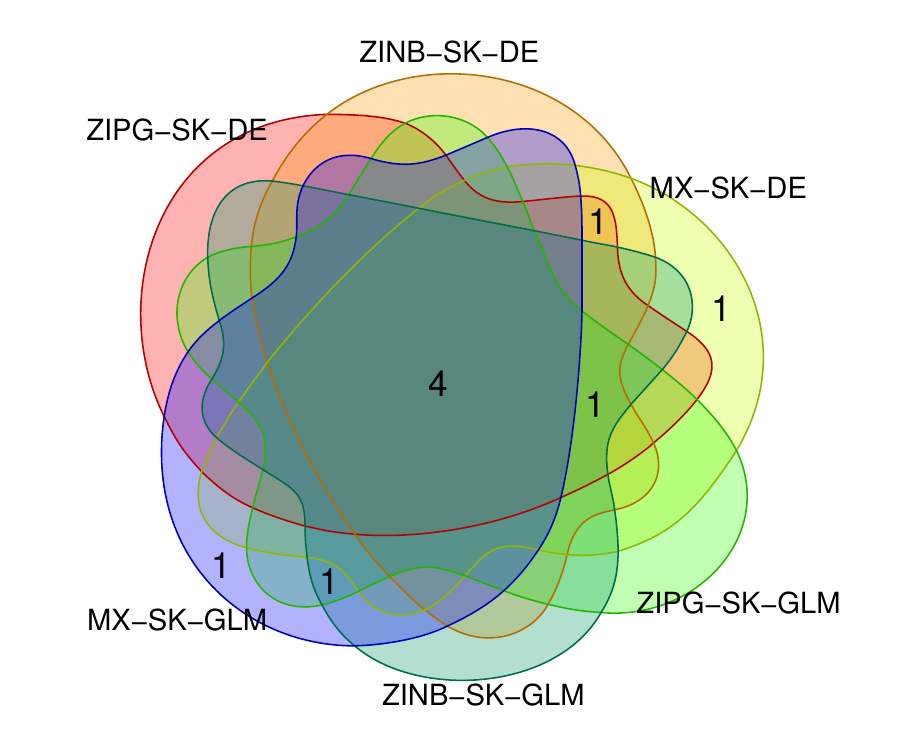} \vspace{4pt}
    \includegraphics[width=1\linewidth]{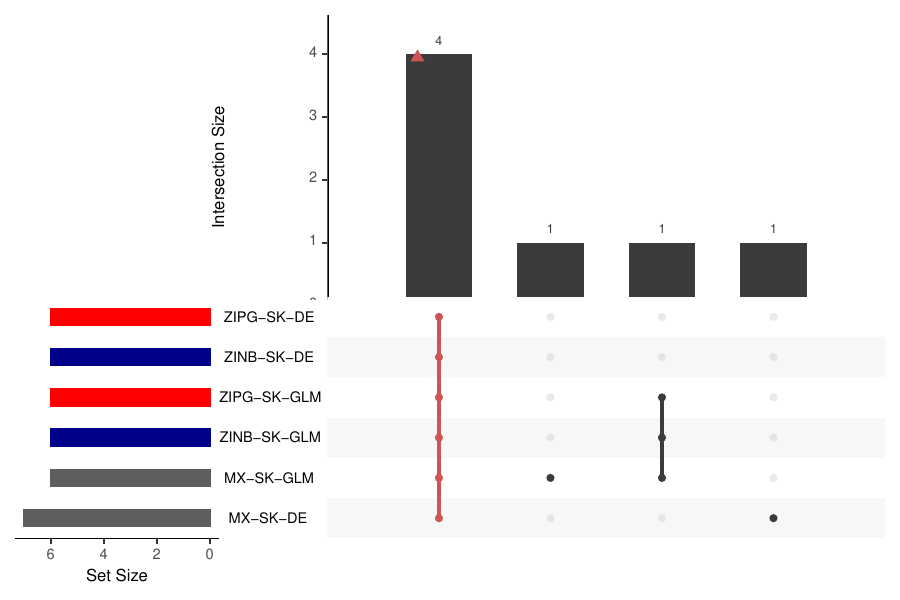}
    \end{minipage}}
    \subfloat[Other methods]{
    \begin{minipage}[c]{0.22\textwidth}
    \includegraphics[width=1\linewidth]{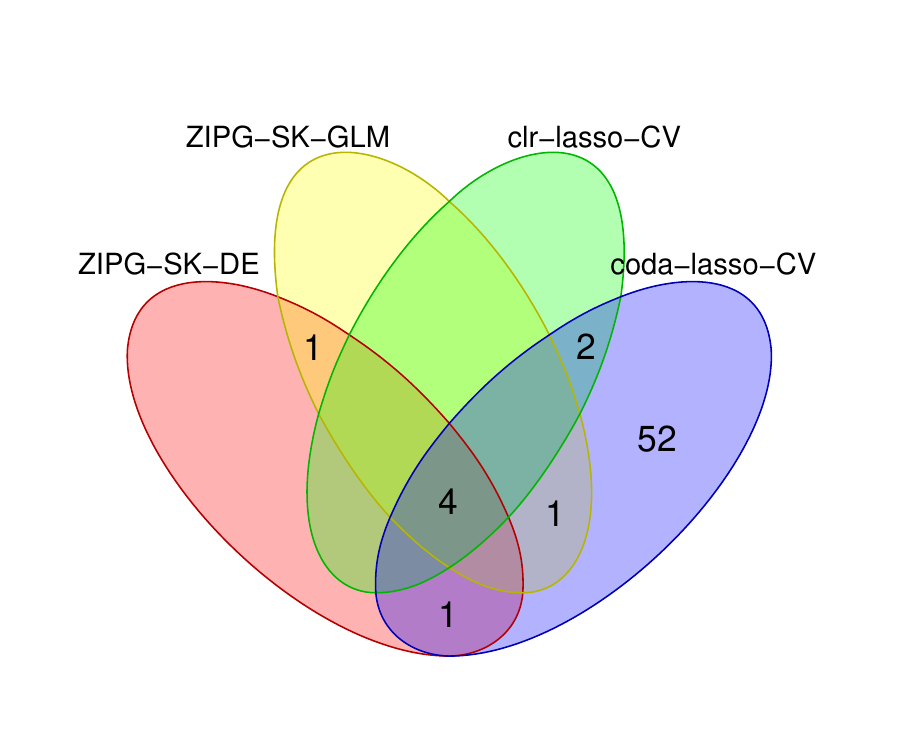} \vspace{4pt}
    \includegraphics[width=1\linewidth]{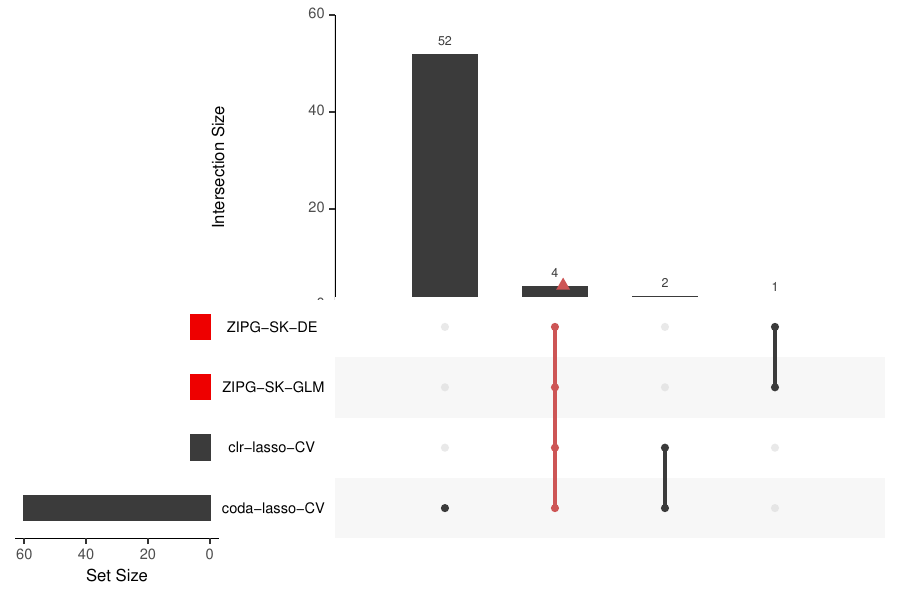}
    \end{minipage}}
    \caption{ 
    The number and overlap of selected gut microbial variables under different methods and strategies. Panel (a) displays the results obtained using various method strategies explored in this paper. ZIPG, ZINB, and MX represent the data generation methods, while DE, GLM, and RF represent the methods used for test statistic calculation. The Inter, Pool, and SK signify the Intersection, Pooling, and Simultaneous strategies of variable set selection, respectively. Panel (b) presents the results obtained using the Simultaneous strategy. Panel (c) presents a comparison between the final method (ZIPG-SK-DE, ZIPG-SK-GLM) and the two penalty regression methods discussed in this paper.}
    \label{fig_real}
\end{figure}
\vspace{-6pt}


The ZIPG-SK methods successfully identify six gut microbes at the species level associated with CRC across various regions. Five of these identified species exhibited consistent overlap across methods: \textit{Fusobacterium nucleatum} (F. nucleatum), \textit{Gemella morbillorum} (G. morbillorum), \textit{Parvimonas micra} (P. micra), \textit{Parvimonas unclassified} (P. unclassified), and \textit{Peptostreptococcus stomatis} (P. stomatis). Among these, \textit{F. nucleatum}, \textit{G. morbillorum}, and \textit{P. micra} are consistent with findings reported by \citet{liu2022multi}. While we note that \citet{liu2022multi}'s feature selection, based on pooling, resulted in the selection of more gut microbes, our findings for \textit{P. stomatis} and \textit{F. nucleatum} align with \citet{shuwen2022using}. The consistent selection of \textit{P. unclassified}, despite lacking similar results in existing literature, suggests a potential guiding role for further etiological research.

\vspace{-3pt}
Further comparison (Figure~\ref{fig_real} (c)) is made between the ZIPG-SK methods (ZIPG-SK-DE and ZIPG-SK-GLM) and two penalized regression benchmarks, clr-lasso and coda-lasso. Clr-lasso selected 6 gut microorganisms, while coda-lasso selected 60 microorganisms, with an overlap of 4 microorganisms across the four methods. We performed a classification diagnosis of CRC based on the gut microbial selection results from the four methods (Table S.2). Although the method proposed in this paper selects relatively fewer gut microbes, the classification accuracy of all four methods is comparable, confirming the strong efficacy of the presented method in subsequent diagnostic tasks. In summary, the proposed method effectively identifies common gut microbial markers from diverse sources using multi-country gut microbial datasets. The analysis above utilizes the Simultaneous method in the OSFF mode of Direct Diff Dot product, while the results based on other OSFF cases are discussed in Supplementary Materials S4.4.

\vspace{-3pt}
\textit{Robustness and Generalizability.} To assess the ZIPG-SK method's robustness, we applied it to the T2D dataset, successfully identifying differentially expressed genes (DEGs) across six distinct cell types, including known markers ($\textit{INS}$, $\textit{CLPS}$ \citep{martinez2022effect}), and selecting a more concise set of DEGs compared to clr-lasso and coda-lasso (Figure S.11). This confirms the framework's improved ability to capture shared biological signals from highly heterogeneous data.





\vspace{-24pt}
\section{Conclusion} \label{s:conclusion}
\vspace{-6pt}

The ZIPG-SK framework represents a versatile and rigorous solution for variable selection in multi-source count data. Our core contribution lies in utilizing the Gaussian copula and ZIPG distribution to accurately model and synthesize complex count features while effectively incorporating covariate information within the multi-source knockoff framework. By leveraging the general knockoff method for statistical calculation and FDR control, the ZIPG-SK approach successfully addresses the conditional association testing problem across heterogeneous datasets.
Simulation results and real-data applications on gut microbiome and single-cell RNA-seq data consistently demonstrate the superior performance of our proposed method in controlling the FDR and improving power compared to existing Simultaneous Knockoff methods designed for count data. In the real-data context, ZIPG-SK successfully identified a concise and biologically relevant set of core features, validating its efficacy in capturing robust signals from diverse sources. 

\vspace{-3pt}
The current work suggests several compelling avenues for future research and extension. For instance, our analysis (Figure S.6) indicates that a limited number of clinical factors may not substantially influence the final variable selection results in high-dimensional settings, possibly due to low correlation with the true microbial signals. Future work could thus focus on developing specialized knockoff constructions that more effectively leverage high-dimensional or complex covariate structures. 
Furthermore, the ZIPG-SK methodology is inherently flexible and can be readily adapted to scenarios involving continuous or mixed response variables. 
Finally, since count features often exhibit hierarchical or inherent grouping structures (e.g., family or species levels), extending the method to perform group-level or multi-level variable selection would be valuable, potentially building upon frameworks like the Multilayer Knockoff Filter (MKF) \citep{katsevich2019multilayer}, to delve deeper into the selection of count feature groups. 
Our ZIPG-SK framework may serve as a robust foundation for principled variable selection in complex multi-source count data analysis.

\backmatter

\vspace{-36pt}
\section*{Acknowledgements}
This work was supported by the National Natural Science Foundation of China [12401381], and National Social Science Foundation Project [20BTJ035]. \vspace*{-8pt}

\vspace{-18pt}
\section*{Data and Code Availability}
\vspace*{-6pt}
The metagenomic sequencing data are available in the Sequence Read Archive (\href{https://www.ncbi.nlm.nih.gov/sra}{SRA}) and the European Nucleotide Archive (\href{https://www.ncbi.nlm.nih.gov/sra}{SRA}). The single-cell diabetes dataset is sourced from the European Molecular Biology Laboratory-European Bioinformatics Institute (\href{https://www.ebi.ac.uk/gxa/sc/experiments/E-MTAB-5061/downloads}{EMBL-EBI}). An R package ZIPG-SK implementing our method are available from \href{https://github.com/tsnm1/ZIPGSK}{https://github.com/tsnm1/ZIPGSK}.

\vspace{-24pt}
\bibliographystyle{biom} 
\bibliography{biomsample_bib}

\vspace{-24pt}
\section*{Supporting Information}
\vspace*{-12pt}
The supporting materials is available with this paper at the Biometrics website on Wiley Online Library.\vspace*{-6pt}

\label{lastpage}
\end{document}